\documentclass[11pt,letterpaper]{amsart} % 11.23 use amsart 

\usepackage{amsmath,amsfonts,amsthm,amssymb,stmaryrd,bm,cite,enumerate}
\theoremstyle{plain}

\numberwithin{equation}{section}

\newtheorem{thm}{Theorem}[section]
\newtheorem{lem}[thm]{Lemma}

\allowdisplaybreaks  % introduced 7.15.15

\newcounter{cond}

\newcommand{\real}{{\mathbb R}}

\newcommand{\trace}{\mathrm{tr\,}}
\newcommand{\rmin}{\mathrm{In\,}}
\newcommand{\rmob}{\mathrm{Ob\,}}
\newcommand{\rmsub}{\mathrm{Sub}}
\newcommand{\ityes}{\textit{yes}}      
\newcommand{\itno}{\textit{no}}

\newcommand{\tbullet}{\mathrel{\raise .2ex\hbox{\tiny$\bullet$}}} % 5.8.20 THIS for larger cdot as times

\newcommand{\dscript}{\mathcal{D}}
\newcommand{\escript}{\mathcal{E}}
\newcommand{\hscript}{\mathcal{H}}
\newcommand{\iscript}{\mathcal{I}}
\newcommand{\jscript}{\mathcal{J}}

\newcommand{\lscript}{\mathcal{L}}
\newcommand{\oscript}{\mathcal{O}}
\newcommand{\sscript}{\mathcal{S}}
\newcommand{\iscriptbar}{\overline{\iscript}}
\newcommand{\iscripthat}{\widehat{\iscript}}
\newcommand{\jscripthat}{\widehat{\jscript}}
\newcommand{\jhat}{\widehat{J}}

\newcommand{\stilde}{\widetilde{s}}

\newcommand{\ab}[1]{\left|#1\right|}
\newcommand{\brac}[1]{\left\{#1\right\}}
\newcommand{\paren}[1]{\left(#1\right)}
\newcommand{\sqbrac}[1]{\left[#1\right]}

\newcommand{\elbows}[1]{{\left\langle#1\right\rangle}}
\newcommand{\ket}[1]{{\left|#1\right>}}
\newcommand{\bra}[1]{{\left<#1\right|}}

\begin{document}

\title{QUANTUM TRANSITION PROBABILITIES}

\author{Stan Gudder}
\address{Department of Mathematics, 
University of Denver, Denver, Colorado 80208}
\email{sgudder@du.edu}
\date{}
\maketitle

\begin{abstract}
Transition probabilities are an important and useful tool in quantum mechanics. However, in their present form, they are limited in scope and only apply to pure quantum states. In this article we extend their applicability to mixed states and to transitions between quantum effects. We also present their dependence on a measured operation or instrument. We begin by defining our concepts on a general quantum effect algebra. These concepts are illustrated using Holevo operations and instruments. We then present transition probabilities in the special case of the Hilbert space formulation of quantum mechanics. We show that for pure states and particular types of operations the transition probabilities reduce to their usual form. We give examples in terms of L\"uders operations and instruments.
\end{abstract}

\section{Introduction}  % Section 1
In the usual Hilbert-space formulation of quantum mechanics, an important and useful role is played by transition probabilities \cite{hz12,nc00}. For example, suppose a quantum system is described by a Hilbert space $H$ and the system is in a pure state given by a unit vector $\psi\in H$. We now measure the energy of the system which is described by a self-adjoint operator $A$ on $H$. The energy values are given by eigenvalues $\lambda _i$ of $A$ with corresponding unit eigenvectors $\phi _i$ so
$A\phi _i=\lambda _i\phi _i$, $i=1,2,\ldots\,$. According to the axioms of quantum mechanics, the probability that the energy of the system is $\lambda _i$ becomes
$P(\psi ,\phi _i)=\ab{\elbows{\psi ,\phi _i}}^2$ \cite{hz12,nc00}. We call $P(\psi ,\phi _i)$ the \textit{transition probability} of $\psi$ to $\phi _i$. This definition is appropriate because $\brac{\phi _i\colon i=1,2,\ldots}$ is an orthonormal basis for $H$ so we have $\sum\limits _i\ab{\elbows{\psi ,\phi _i}}^2=1$. Also, another axiom of quantum mechanics says that if the energy value observed is $\lambda _i$ then the state $\psi$ updates to the state $\phi _i$ \cite{hz12,nc00}. Thus, $P(\psi ,\phi _i)$ is the probability that this update occurs. In general, if $\psi ,\phi\in H$ specify two states then calling $P(\psi ,\phi )=\ab{\elbows{\psi ,\phi}}^2$ the transition probability of $\psi$ to $\phi$ does not stress the fact that this transition is caused by a measurement of an observable $A$. Also, what about mixed states described by density operators on $H$? Moreover, what about more general measurements such as operations and instruments \cite{bgl95,dl70,kra83}?

In Section 2, we study probabilistic models that generalize quantum mechanics. We begin by introducing an effect algebra $E$ \cite{bgl95,dl70}. An effect $a\in E$ is a
\ityes-\itno\ (true-false) experiment that either occurs (\ityes ) or does not occur (\itno ) when measured. For $a,b\in E$ there is a partial operation $a\oplus b$ that describes a statistical sum of $a,b$ (when it is defined). A \textit{substate} on $E$ is function $s\colon E\to\sqbrac{0,1}\subseteq\real$ that satisfies $s(a\oplus b)=s(a)+s(b)$ when
$a\oplus b$ is defined. There is a special effect $1\in E$ that is always true and a substate $s$ is a \textit{state} if $s(1)=1$. We denote the set of states on $E$ by
$\sscript (E)$ and the set of substates by $\rmsub (E)$. If $s\in\sscript (E)$, then $s(a)$ is the probability that $a$ occurs when the system is in state $s$. A convex map
$\iscript\colon\rmsub (E)\to\rmsub (E)$ is called an operation \cite{bgl95,dl70}. An operation corresponds to an apparatus that can be employed to measure an effect and we denote the set of operations by $\oscript (E)$. If $s\in\sscript (E)$, $\iscript\in\oscript (E)$ with $\iscript (s)\ne 0$, $a,b\in E$ we define the
\textit{transition probability of $a$ to $b$ relative to} $s,\iscript$ by
\begin{equation*}
P_{s,\iscript}(a,b)=\frac{s(a)}{\iscript (s)(1)}\,\iscript (s)(b)
\end{equation*}
If $\iscript\in\oscript (E)$, $s_1,s_2\in\sscript (E)$ we define the \textit{transition probability of $s_1$ to $s_2$ relative to} $\iscript$ by
\begin{equation*}
P_\iscript (s_1,s_2)=\iscript (s_1)(1)\iscript\paren{\iscript (s_2)(1)}
\end{equation*}
Section 2 derives properties of $P_{s,\iscript}(a,b)$ and $P_\iscript (s_1,s_2)$. We also define repeatable operations and discuss their properties. Moreover, Section 2 discusses generalizations of effects and operations called observables and instruments, respectively. An observable is an effect-valued measure $A_x\in E$ where $A_x$ is the effect that occurs when the observable has the outcome $x$ upon being measure \cite{gud20,gud22}. An instrument is an operation-valued measure
$\iscript _x\in\oscript (E)$ where $\iscript _x$ is the operation that occurs when the instrument has the outcome $x$ upon being measured \cite{gud20,gud22,gud23}. We then have the transition probabilities $P_{\iscript _x}(s_1,s_2)$. If $A_x,B_y$ are observables, we have the transition probabilities $P_{s,\iscript _z}(A_x,B_y)$.

Section 2 ends with a discussion of Holevo operations and instruments which are used to illustrate the previously introduced concepts \cite{hol82}. If
$\alpha\in\sscript (E)$, $a\in E$, a pure Holevo operation has the form $\hscript ^{(a,\alpha )}(s)=s(a)\alpha$. We also define mixed Holevo operations and Holevo instruments in a natural way \cite{hol82}. Transition probabilities relative to these operations and instruments are computed. We also find updated states after these operations and instruments are measured.

Section 3 specializes the theory of effect algebras to the traditional Hilbert space formulation of quantum mechanics. In this formulation, effects are described by operators $A$ satisfying $0\le A\le I$ and states are described by effects $\rho$ with trace $\trace (\rho )=1$. We call $\rho$ a \textit{density operator} and the corresponding state satisfies $s(A)=\trace (\rho A)$. In this formalism an operation is a completely positive trace nonincreasing linear map from the set of density operators to the set of positive operators on the Hilbert space $H$ \cite{bgl95,dl70,hz12, nc00}. We point out that every operation has a Kraus decomposition $\iscript (\rho )=\sum K_i\rho K_i^*$ where
$K_i$ are linear operators on $H$ satisfying $\sum K_i^*K_i\le I$. In particular, we study L\"uders operations which have the form $\iscript (\rho )=A^{1/2}\rho A^{1/2}$ where $A$ is an effect \cite{lud51}. L\"uders instruments are defined in the natural way. We then compute transition probabilities relative to these operations and instruments.

\section{Basic Concepts and Definitions}  % Section 2
An \textit{effect algebra} $\escript =(E,0,1,\oplus )$ \cite{bgl95,gud20,gud22,hz12} is a set of effects $E$ with two special elements $0,1\in E$ and a partial binary operation $\oplus$ satisfying the following conditions:
\begin{list}
{(EA\arabic{cond})}{\usecounter{cond}
\setlength{\rightmargin}{\leftmargin}}
%(EA1)
\item if $a\oplus b$ is defined, then $b\oplus a$ is defined and $b\oplus a=a\oplus b$,
%(EA2)
\item if $a\oplus (b\oplus c)$ is defined, then $(a\oplus b)\oplus c$ is defined and
\begin{equation*}
(a\oplus b)\oplus c=a\oplus (b\oplus c)
\end{equation*}
%(EA3)
\item for every $a\in E$, there exists a unique effect $a'\in E$ such that $a\oplus a'=1$,
%(EA4)
\item if $a\oplus 1$ is defined, then $a=0$.
\end{list}

An effect corresponds to a \ityes-\itno\ experiment that either occurs or does not occur when it is measured. The effect $0$ never occurs and the effect $1$ always occurs when measured. The \textit{complement} $a'$ of $a$ occurs if and only if $a$ does not occur. The sum $a\oplus b$ corresponds to a statistical sum of $a$ and $b$ when it is defined. A \textit{substate} on $\escript$ is a function $s\colon E\to\sqbrac{0,1}\subseteq\real$ that satisfies $s(a\oplus b)=s(a)+s(b)$ whenever $a\oplus b$ is defined. A substate $s$ that satisfies $s(1)=1$ is called a \textit{state}. We denote the set of substates by $\rmsub (\escript )$ and the set of states by $\sscript (\escript )$. A state $s$ describes the initial condition of a physical system and $s(a)$ is the probability that $a$ occurs when the system is in the state $s$. If $s\in\sscript (\escript )$,
$\lambda\in\sqbrac{0,1}$, then $\lambda s$ defined by $(\lambda s)(a)=\lambda s(a)$ is a substate. Conversely, if $s\in\rmsub (\escript )$ and $s\ne 0$, then
$\stilde (a)=s(a)/s(1)$ is a state.

A function $J\colon\rmsub (\escript )\to\rmsub (\escript )$ that satisfies $J\paren{\sum\limits _{i=1}^n\lambda _is_i}=\sum\limits _{i=1}^n\lambda _iJ(s)$ when
$\sum\limits _{i=1}^n\lambda _i\le 1$ is called a \textit{convex function} or an \textit{operation} \cite{bgl95,dl70,kra83}. We denote the set of operations on $\escript$ by
$\oscript (\escript )$ and it follows that if $J\in\oscript (\escript )$ then $J(\lambda s)=\lambda J(s)$ for all $\lambda\in\sqbrac{0,1}$, $s\in\rmsub (\escript )$. If
$J\in\oscript (\escript )$ satisfies $J(s)\in\sscript (\escript )$ for all $s\in\sscript (\escript )$, then $J$ is a \textit{channel} \cite{bgl95,hz12, nc00}. We say that
$J\in\oscript (\escript )$ \textit{measures} $a\in E$ if $s(a)=J(s)(1)$ for all $s\in\sscript (\escript )$. We say that $\sscript (\escript )$ is \textit{separating} if
$s(a)=s(b)$ for all $s\in\sscript (\escript )$ implies $a=b$. If $\sscript (\escript )$ is separating, then $J$ measures the unique effect $\jhat$ satisfying $s(\jhat\,)=J(s)(1)$ for all $s\in\sscript (\escript )$. We then consider $J$ as an apparatus that measures the effect $\jhat$. As with effects, an operation is considered to be a \ityes-\itno\ measurement. If $s\in\sscript (\escript )$, then $J(s)(1)$ is the \textit{probability $J$ occurs} when the system is in the state $s$. If $J(s)\ne 0$, we call
$J(s)^\sim =J(s)/J(s)(1)\in\sscript (\escript )$ the \textit{updated state} after $J$ is measured where $0\in\rmsub (\escript )$ means $0(a)=0$ for all $a\in E$.

If $s\in\sscript (\escript )$, $J\in\oscript (\escript )$ with $J(s)\ne 0$ and $a,b\in E$, the \textit{transition probability of $a$ then $b$ relative to} $s,J$ is 
\begin{equation*}
P_{s,J}(a,b)=\frac{s(a)}{J(s)(1)}\,J(s)(b)=s(a)J(s)^\sim (b)
\end{equation*}
Thus, $P_{s,J}(a,b)$ is the probability that $a$ occurs in the state $s$ times the probability that $b$ occurs in the updated state $J(s)^\sim$ after $J$ is measured. If
$\iscripthat =a$, then $P_{s,\iscript}(a,b)=\iscript (s)(b)$. Notice that $P_{s,J}(a,\tbullet )\in\rmsub (\escript )$ and $P_{s,J}(\tbullet ,b)\in\rmsub (\escript )$. Moreover,
$P_{s,J}(a,1)=s(a)$ for all $a\in E$, $P_{s,J}(1,b)=\iscript (s)^\sim (b)$ for all $b\in E$ and $P_{s,J}(1,1)=1$.

If $f,g\colon S\to S$ are functions, their \textit{composition} $f\circ g\colon S\to S$ is $f\circ g(t)=f\paren{g(t)}$. We say that $J\in\oscript (\escript )$ is
\textit{repeatable} if $J\circ J=J$. If $J\in\oscript (\escript )$, $s_1,s_2\in\sscript (\escript )$, the \textit{transition probability of $s_1$ to $s_2$ relative to } $J$ is
\begin{equation*}
P_J(s_1,s_2)=J(s_1)(1)J\circ J(s_2)(1)=J\sqbrac{J(s_1)(1)J(s_2)}(1)
\end{equation*}
Thus, $P_\jscript (s_1,s_2)$ is the probability $J$ occurs when the system is in state $s_1$ times the probability the composition $J\circ J$ occurs when the system is in state $s_2$. We see that if $J$ is repeatable, then
\begin{equation*}
P_J(s_1,s_2)=J(s_1)(1)J(s_2)(1)
\end{equation*}
in which case $P_J(s_1,s_2)=P_J(s_2,s_1)$ which does not happen in general. Moreover, if $J$ also measures $a$, then $P_J(s_1,s_2)=s_1(a)s_2(a)$. The following lemma shows the relationship between these two types of transition probabilities. We denote the product of two real-valued functions $f,g$ by $f\tbullet g(x)=f(x)g(x)$.

\begin{lem}    % Lemma 2.1
\label{lem21}
If $\sscript (\escript )$ is separating, then $P_J(s_1,s_2)=P_{s_1,J}\tbullet P_{J(s_2),J}(\jhat ,1)$.
\end{lem}
\begin{proof}
Since
\begin{equation*}
s_2\sqbrac{(J\circ J)^\wedge}=(J\circ J)(s_2)(1)=J\sqbrac{J(s_2)}(1)=P_{J(s_2),J}(\jhat ,1)
\end{equation*}
we obtain
\begin{align*}
P_J(s_1,s_2)&=s_1(\jhat\,)s_2\sqbrac{(J\circ J)^\wedge}=P_{s_1,J}(\jhat ,1)P_{J(s_2),J}(\jhat ,1)\\
   &=P_{s_1,J}\tbullet P_{J(s_2),J}(\jhat ,1)\qedhere
\end{align*}
\end{proof}

An \textit{observable} on $\escript$ is a finite set of effects $A=\brac{A_x\colon x\in\Omega _A}$, $A_x\in E$, such that $\sum\limits _{x\in\Omega _A}A_x=1$. We call
$\Omega _A$ the \textit{outcome space} of $A$ and $A_x$ is the effect that occurs when $A$ has outcome $x\in\Omega _A$. For $\Delta\subseteq\Omega _A$ we define the \textit{effect-valued measure} $A(\Delta )=\sum\limits _{x\in\Delta}A_x$ \cite{bgl95,hz12}. We see that $A(\Omega _A)=1$ and
$A(\Delta _1\cup\Delta _2)=A(\Delta _1)+A(\Delta _2)$ when $\Delta _1\cap\Delta _2=\emptyset$. We denote the set of observables on $\escript$ by
$\rmob (\escript )$. If $s\in\sscript (\escript )$, $A\in\rmob (\escript )$, the \textit{distribution} of $A$ in the state $s$ is the probability measure
$\Phi _s^A(\Delta )=\sum\limits _{x\in\Delta}s(A_x)$ on $\Omega _A$. Let $A,B\in\rmob (\escript )$, $s\in\sscript (\escript )$, $J\in\oscript (\escript )$ with $J(s)\ne 0$ and consider the transition probability
\begin{equation*}
P_{xy}=P_{s,J}(A_x,B_y)=\frac{s(A_x)}{J(s)(1)}\,J(s)(B_y)
\end{equation*}
For $\Delta\subseteq\Omega _A\times\Omega _B$ we define $P(\Delta )=\sum\brac{P_{xy}\colon (x,y)\in\Delta}$.

\begin{lem}    % Lemma 2.2
\label{lem22}
$P$ is a probability measure on $\Omega _A\times\Delta _B$ with marginals
\begin{align*}
P^1(\Delta _1)&=P(\Delta _1\times\Omega _B)=s\sqbrac{A(\Delta _1)}=\Phi _s^A(\Delta _1)\hbox{ where }\Delta _1\subseteq\Omega _A\\
P^2(\Delta _2)&=P(\Omega _A\times\Delta _2)=J(s)^\sim\sqbrac{B(\Delta _2)}=\Phi _{J(s)^\sim}^B(\Delta _2)\hbox{ where }\Delta _2\subseteq\Omega _B
\end{align*}
\end{lem}
\begin{proof}
We have that
\begin{align*}
P(\Omega _A\times\Omega _B)&=\sum\brac{P_{xy}\colon (x,y)\in\Omega _A\times\Omega _B}\\
   &=\frac{1}{J(s)(1)}\sum\brac{s(A_x)J(s)(B_y)\colon (x,y)\in\Omega _A\times\Omega _B}\\
   &=\frac{1}{J(s)(1)}\sum _{x\in\Omega _A}s(A_x)\sum _{y\in\Omega _B}J(s)(B_y)=\frac{1}{J(s)(1)}\,J(s)(1)=1
\end{align*}
If $\Delta _1,\Delta _2\subseteq\Omega _A\times\Omega _B$ with $\Delta _1\cap\Delta _2=\emptyset$, then
\begin{align*}
P(\Delta _1\cup\Delta _2)&=\sum\brac{P_{xy}\colon (x,y)\in\Delta _1\cup\Delta _2}\\
   &=\sum\brac{P_{xy}\colon (x,y)\in\Delta _1}+\sum\brac{P_{xy}\colon (x,y)\in\Delta _2}\\
   &=P(\Delta _1)+P(\Delta _2)
\end{align*}
Hence, $P$ is a probability measure on $\Omega _A\times\Omega _B$. The marginals of $P$ become
\begin{equation*}
P^1(\Delta _1)=P(\Delta _1\times\Omega _B)=\sum _{(x,y)\in\Delta _1\times\Omega _B}\frac{s(A_x)}{J(s)(1)}\,J(s)(B_y)
   =s\sqbrac{A(\Delta _1)}=\Phi _s^A(\Delta _1)
\end{equation*}
where $\Delta _1\subseteq\Omega _A$ and 
\begin{align*}
P^2(\Delta _2)&=P(\Omega _A\times\Delta _2)=\sum _{(x,y)\in\Omega _A\times\Delta _2}\frac{s(A_x)}{J(s)(1)}\,J(s)(B_y)\\
   &=J(s)^\sim\sqbrac{B(\Delta _2)}=\Phi _{J(s)^\sim}^B(\Delta _2)
\end{align*}
where $\Delta _2\subseteq\Omega _B$.
\end{proof}

An \textit{instrument} $\iscript$ on $\escript$ is a finite set of operations $\iscript =\brac{\iscript _x\colon x\in\Omega _\iscript}$, $\iscript _x\in\oscript (\escript )$ such that
$\sum\limits _{x\in\Omega _\iscript}\iscript _x$ is a channel \cite{bgl95,dl70,gud23,hz12}. We call $\Omega _\iscript$ the \textit{outcome space} of $\iscript$. If
$s\in\sscript (\escript )$, then $\iscript _x(s)(1)$ is the \textit{probability} that the outcome $x$ occurs when $\iscript$ is measured and the system is in state $s$. Moreover, we call $\Phi _s^\iscript (\Delta )=\sum\limits _{x\in\Delta}\iscript _x(s)(1)$ the \textit{distribution} of $\iscript$ in the state $s$. Since
$\iscriptbar =\sum _{x\in\Omega _\iscript}\iscript _x$ is a channel we have
\begin{equation*}
\Phi _s^\iscript (\Omega _\iscript )=\sum _{x\in\Omega _\iscript}\iscript _x(s)(1)=\sqbrac{\sum _{x\in\Omega _\iscript}\iscript _x(s)}(1)=1
\end{equation*}
As in Lemma~\ref{lem22}, if $\Delta _1,\Delta _2\subseteq\Omega _\iscript$ and $\Delta _1\cap\Delta _2=\emptyset$, then 
\begin{equation*}
\Phi _s^\iscript (\Delta _1\cup\Delta _2)=\Phi _s^\iscript (\Delta _1)+\Phi _s^\iscript (\Delta _2)
\end{equation*}
so $\Phi _s^\iscript$ is a probability measure on $\Omega _\iscript$. We denote the set of instruments on $\escript$ by $\rmin (\escript )$. If $\iscript\in\rmin (\escript )$, we say that $\iscript$ \textit{measures} an observable $A\in\rmob (\escript )$ with $\Omega _A=\Omega _\iscript$ if $\iscript _x(s)(1)=s(A_x)$ for all $s\in\sscript (\escript )$
$x\in\Omega _\iscript$. If $\escript$ has a separating set of states, $A$ is unique and we write $A=\iscripthat$.

If $\iscript\in\rmin (\escript )$ and $A,B\in\rmob (\escript )$ we obtain the transition probability
\begin{equation*}
P_{s,\iscript _z}(A_x,B_y)=\frac{s(A_x)}{\iscript _z(s)(1)}\,\iscript _z(s)(B_y)
\end{equation*}
Multiplying by $\iscript _z(s)(1)$, we have the function on $\Omega _\iscript\times\Omega _A\times\Omega _B$ given by
\begin{equation*}
Q_{s,\iscript _z}(A_x,B_y)=\iscript (s)(1)P_{s,\iscript _z}(A_x,B_y)=s(A_x)\iscript _x(s)(B_y)
\end{equation*}
Since
\begin{equation*}
\sum\brac{Q_{s,\iscript _z}(A_xB_y)\colon(z,x,y)\in\Omega _\iscript\times\Omega _A\times\Omega _B}=s(1)\iscriptbar (s)(1)=1
\end{equation*}
we conclude that $Q$ determines a probability measure on $\Omega _\iscript\times\Omega _A\times\Omega _B$. The marginals of $Q$ become
\begin{align*}
Q^2(x)&=s(A_x)\sum\brac{\iscript _z(s)(B_y)\colon (z,y)\in\Omega _\iscript\times\Omega _B}\\
   &=s(A_x)\iscriptbar (s)(1)=s(A_x)
\end{align*}
so $Q^2(\Delta )=\Phi _s^A(\Delta )$, $\Delta\subseteq\Omega _A$,
\begin{align*}
Q^3(y)&=\sum\brac{s(A_x)\iscript _z(s)(B_y)\colon (z,x)\in\Omega _\iscript\times\Omega _A}\\
   &=\iscriptbar (s)(B_y)
\end{align*}
so $Q^3(\Delta )=\Phi _{\iscriptbar (s)}^B(\Delta )$, $\Delta\subseteq\Omega _B$,
\begin{align*}
Q^1(z)&=\sum\brac{s(A_x)\iscript _z(s)(B_y)\colon (x,y)\in\Omega _A\times\Omega _B}\\
   &=\iscript _z(s)(1)
\end{align*}
so $Q^1(\Delta )=\Phi _s^\iscript (\Delta )$, $\Delta\subseteq\Omega _\iscript$. We conclude that $Q^1$, $Q^2$, $Q^3$ give the distributions of $\iscript$ in the state $s$, $A$ in the state $s$ and $B$ in the state $\iscriptbar (s)$.

If $\iscript\in\rmin (\escript )$, we define the \textit{transition probability} from $s_1$ to $s_2$ relative to $\iscript$ and $(x,y)\in\Omega _\iscript\times\Omega _\iscript$ by 
\begin{equation*}
(P_\iscript )_{(x,y)}(s_1,s_2)=\iscript _x(s_1)(1)\iscript _y\paren{\iscriptbar (s_2)}(1)
\end{equation*}
We see that $P_\iscript$ is a probability measure on $\Omega _\iscript\times\Omega _\iscript$ because 
\begin{equation*}
\sum _{x,y}(P_\iscript )_{(x,y)}(s_1,s_2)=\iscriptbar (s_1)(1)\iscriptbar\paren{\iscriptbar (s_2)}(1)=1
\end{equation*}
The marginals of $P_\iscript$ become
\begin{equation*}
\sum _{y\in\Omega _\iscript}(P_\iscript )_{(x,y)}(s_1,s_2)=\iscript _x(s_1)(1)
\end{equation*}
which is the probability of outcome $x$ when $\iscript$ is measured in the state $s_1$
\begin{equation*}
\sum _{x\in\Omega _\iscript}(P_\iscript )_{(x,y)}(s_1,s_2)=\iscript _y\iscriptbar (s_2)(1)
\end{equation*}
which is the probability of outcome $y$ when $\iscript$ is measured in the state $\iscriptbar (s_2)$.

We now illustrate the previous theory with a particular type of operation and instrument. If $\alpha\in\sscript (\escript )$, $a\in E$, we call
$\hscript ^{(a,\alpha )}\in\oscript (E)$ defined by $\hscript ^{(a,\alpha )}(s)=s(a)\alpha$ a \textit{pure Holevo operation} \cite{hol82}. Since
$\sqbrac{\hscript ^{(a,\alpha )}(s)}(1)=s(a)$ for all $s\in\sscript (\escript )$ we see that $\hscript ^{(a,\alpha )}$ measures $a$. If $A\in\rmob (\escript )$ and
$\alpha _x\in\sscript (\escript )$ for $x\in\Omega _A$ we call
\begin{equation*}
\hscript ^{(A,\brac{\alpha})}(s)=\sum _{x\in\Omega _A}s(A_x)\alpha _x
\end{equation*}
a \textit{mixed Holevo operation} \cite{hol82}. For $s\in\sscript (\escript )$ we have
\begin{equation*}
\sum _{x\in\Omega _A}s(A_x)=s\paren{\sum _{x\in\Omega _A}A_x}=s(1)=1
\end{equation*}
so $\hscript ^{\paren{A,\brac{\alpha}}}$ is a channel. We also call
\begin{equation*}
\hscript _x^{\paren{A,\brac{\alpha}}}(s)=s(A_x)\alpha _x
\end{equation*}
a \textit{Holevo instrument} \cite{hol82}. Notice that $\hscript _x^{\paren{A,\brac{\alpha}}}$ is an instrument because
\begin{equation*}
\sum _{x\in\Omega _A}\sqbrac{\hscript _x^{\paren{A,\brac{\alpha}}}(s)}(1)=\sum _{x\in\Omega _A}s(A_x)=1
\end{equation*}
Also, $\hscript _x^{\paren{A,\brac{\alpha}}}$ measures $A$ because $\sqbrac{\hscript _x^{\paren{A,\brac{\alpha}}}(s)}(1)=s(A_x)$ for all $s\in\sscript (\escript )$ and
$x\in\Omega _A$.

For a pure Holevo operation $\hscript ^{(a,\alpha )}$ and $b,c\in E$, if $s(a)\ne 0$ we have
\begin{equation*}
P_{s,\hscript ^{(a,\alpha )}}(b,c)=\frac{s( b)}{\sqbrac{\hscript ^{(a,\alpha )}(s)}(1)}\sqbrac{\hscript ^{(a,\alpha )}(s)}(c)=\frac{s(b)}{s(a)}\,s(a)\alpha (c)=s(b)\alpha (c)
\end{equation*}
We see that $P_{s,\hscript ^{(a,\alpha )}}$ is independent of $a$ if $s(a)\ne 0$. For a mixed Holevo operation we have
\begin{equation*}
\sqbrac{\hscript ^{\paren{A,\brac{\alpha}}}(s)}(c)=\sum _{x\in\Omega _A}s(A_x)\alpha _x(c)
\end{equation*}
so that
\begin{equation*}
P_{s,\hscript ^{\paren{A,\brac{\alpha}}}}(b,c)=s(b)\sum _{x\in\Omega _A}s(A_x)\alpha _x(c)
\end{equation*}
The state transition probability becomes 
\begin{align*}
P_{\hscript ^{(a,\alpha )}}(s_1,s_2)&=\hscript ^{(a,\alpha )}(s_1)(1)\hscript ^{(a,\alpha )}\sqbrac{\hscript ^{(a,\alpha )}(s_2)}(1)\\
   &=s_1(a)\hscript ^{(a,\alpha )}\sqbrac{s_2(a)\alpha}(1)=s_1(a)s_2(a)\sqbrac{\hscript ^{(a,\alpha )}(\alpha )}(1)\\
   &=s_1(a)s_2(a)\alpha (a)
\end{align*}
Since a mixed Holevo operation is a channel, we have $P_{\hscript ^{\paren{A,\brac{\alpha}}}}(s_1,s_2)=1$ for all $s_1,s_2\in\sscript (\escript )$. For Holevo instruments we have
\begin{align*}
P_{s,\hscript _x^{\paren{A,\brac{\alpha}}}}(b,c)&=\frac{s(b)}{\sqbrac{\hscript _x^{\paren{A,\brac{\alpha}}}(s)}(1)}\sqbrac{\hscript _x^{\paren{A,\brac{\alpha}}}(s)}(c)
   =\frac{s(b)}{s(A_x)}\,s(A_x)\alpha _x(c)\\
   &=s(b)\alpha _x(c)
\intertext{and}
P_{\hscript _x^{\paren{A,\brac{x}}}}(s_1,s_2)&=\hscript _x^{\paren{A,\brac{\alpha}}}(s_1)(1)\hscript _x^{\paren{A,\brac{x}}}
   \sqbrac{\hscript _x^{\paren{A,\brac{\alpha}}}(s_2)}(1)\\
   &=s_1(A_x)\hscript _x^{\paren{A,\brac{\alpha}}}\sqbrac{s_2(A_x)\alpha _x}(1)=s_1(A_x)s_2(A_x)\alpha _x(A_x)
\end{align*}

An operation $\iscript\in\oscript (\escript )$ is \textit{repeatable} if $\iscript\paren{\iscript (s)}=\iscript (s)$ for all $s\in\sscript (\escript )$. Thus, $\iscript\circ\iscript =\iscript$ so we can repeat $\iscript$ without changing the state. If $\iscript$ is repeatable, the transition probability satisfies
\begin{equation*}
P_\iscript (s_1,s_2)=\iscript (s_1)(1)\iscript\paren{\iscript (s_2)}(1)=\iscript (s_1)(1)\iscript (s_2)(1)=P_\iscript (s_2,s_1)
\end{equation*}
We now show that the converse does not hold. We have seen that 
\begin{equation*}
P_{\hscript ^{(a,\alpha )}}(s_1,s_2)=P_{\hscript ^{(a,\alpha )}}(s_2,s_1)
\end{equation*}
However, if $s(a)\ne 0,1$ we obtain
\begin{align*}
\hscript ^{(a,\alpha )}\paren{\hscript ^{(a,\alpha )}(s)}&=\hscript {(a,\alpha )}\paren{s(a)\alpha}=s(a)\hscript ^{(a,\alpha )}(\alpha )\\
   &=s(a)\alpha (a)\alpha\ne s(a)\alpha =\hscript ^{(a,\alpha )}(s)
\end{align*}
Hence, if $s(a)\ne 0$, then $\hscript ^{(a,\alpha )}$ is repeatable if and only if $s(a)=1$.

\section{Hilbert Space Transition Probabilities}  % Section 3
In this section we consider Hilbert space effects, operations, instruments, states and their transition probabilities. These are the basic concepts for the standard quantum probability theory. Let $H$ be a complex Hilbert space which for simplicity, we take to be finite dimensional. Letting $\lscript _+(H)$ be the set of positive linear operators on $H$, an \textit{effect} is represented by an operator $a\in\lscript _+(H)$ satisfying $0\le a\le I$ where $0,I$ are the zero and identity operators respectively 
\cite{bgl95,dl70,hz12}. The set of effects is denoted by $\escript (H)$. A \textit{density operator} $\rho\in\lscript _+(H)$ satisfies $\trace (\rho )=1$ and we denote the set of density operators on $H$ by $\dscript (H)$. According to Born's rule, every state $s\in\sscript (H)$ has the form $s(a)=\trace (\rho a)$ for a unique
$\rho\in\dscript (H)$. This can actually be proved using Gleason's theorem \cite{hz12} so we will identify $\sscript (H)$ and $\dscript (H)$. In the Hilbert space formulation of quantum mechanics, an \textit{operation} is a completely positive linear map $J\colon\dscript (H)\to\lscript _+(H)$ satisfying $\trace\sqbrac{J(\rho )}\le\trace (\rho )$ for all
$\rho\in\dscript (H)$ \cite{bgl95,dl70,hz12,nc00}. We denote the set of operations $H$ by $\oscript (H)$. If $\trace\sqbrac{J(\rho )}=1$ for all $\rho\in\dscript (H)$, then $J$ is a \textit{channel}.

Denoting the set of linear operators on $H$ by $\lscript (H)$, it can be shown that every $J\in\oscript (H)$ has a \textit{Kraus decomposition}
$J(\rho )=\sum\limits _{i=1}^nK_i\rho K_i^*$ where $K_i\in\lscript (H)$ with $\sum\limits _{i=1}^nK_i^*K_i\le I$ \cite{hz12, kra83,nc00}. Notice that $J$ is a channel if and only if  $\sum\limits _{i=1}^nK_i^*K_i=I$. The operators $K_i$ are called Kraus operators for $J$ and they are not unique. If $J$ has the form $J(\rho )=a^{1/2}\rho a^{1/2}$ where
$a\in\escript (H)$ we call $J$ a \textit{L\"uders operation}. If the Kraus operators are projections, it follows that $K_iK_j=0$ for $i\ne j$ and hence, in this case, $J$ is repeatable. Conversely, if $J(\rho )=a^{1/2}\rho a^{1/2}$ is a repeatable L\"uders operation then
\begin{equation*}
a^{1/2}\rho a^{1/2}=J(\rho )=J\paren{J(\rho )}=a\rho a
\end{equation*}
for every $\rho\in\dscript (H)$. Letting $\rho =\tfrac{1}{n}I$ where $n=\dim H$, we conclude that $a=a^2$ so $a$ is a projection.

In the Hilbert space case, it is well-known that $\sscript (H)$ is separating. It follows that an operation $J\in\oscript (H)$ measures a unique effect $\jscripthat\in\escript (H)$ satisfying $\trace (\rho\jhat\,)=\trace\sqbrac{J(\rho )}$ for all $\rho\in\dscript (H)$. Notice that if $J$ is a channel, then $\jhat =I$. Although an operation measures a unique effect, as we shall see, every effect is measured by many operations. If $J$ has Kraus decomposition
$J(\rho )=\sum K_i\rho K_i^*$, since
\begin{align*}
\trace\sqbrac{J(\rho )}&=\trace\paren{\sum K_i\rho K_i^*}=\sum\trace (K_i\rho K_i^*)=\sum\trace (\rho K_i^*K_i)\\
   &=\trace\paren{\rho\sum K_i^*K_i}
\end{align*}
for all $\rho\in\dscript (H)$, we conclude that $\jhat =\sum K_i^*K_i$. In particular, if $J(\rho )=a^{1/2}\rho a^{1/2}$ is a L\"uders operation, then $\jhat =a$.

An \textit{observable} is given by a finite sequence of effects $A=\brac{A_x\colon x\in\Omega _A}$ where $A_x\in\escript (H)$ with $\sum A_x=I$
\cite{bgl95,dl70,gud20,gud22,hz12}. Then $\Omega _A$ is the outcome space of $A$ and $A$ determines an effect-valued measure
$A(\Delta )=\sum\limits _{x\in\Delta}A_x$, $\Delta\subseteq\Omega _A$. The probability that $A$ has outcome $x$ in the state $\rho\in\dscript (H)$ is $\trace (\rho A_x)$ and the $\rho$-distribution of $A$ is $\Phi _\rho ^A(\Delta )=\sum\limits _{x\in\Delta}\trace (\rho A_x)$. The set of observables on $H$ is denoted by $\rmob (H)$. In a similar way, an \textit{instrument} is given by a finite sequence of operations $\iscript =\brac{J _x\colon x\in\Omega _\iscript}$, $J_x\in\oscript (H)$ where $\sum J_x$ is a channel
\cite{bgl95,dl70,gud20,gud23,hz12}. Then $\Omega _\iscript$ is the outcome space of $\iscript$ and $\iscript$ defines an operation-valued measure
$\iscript (\Delta )=\sum\limits _{x\in\Delta}J_x$, $\Delta\in\Omega _\iscript$. The probability that $\iscript$ has outcome $x$ in the state $\rho\in\dscript (H)$ is
$\trace\sqbrac{J_x(\rho )}$ and the $\rho$-distribution of $\iscript$ is $\phi _\rho ^\iscript (\Delta )=\sum\limits _{x\in\Delta}\trace\sqbrac{J_x(\rho )}$. We denote the set of instruments on $H$ by $\rmin (H)$. An instrument $\iscript$ measures the unique observable $\jhat _x$. If $A\in\rmob (H)$, the corresponding L\"uders instrument is
$\lscript _x^A(\rho )=A_x^{1/2}\rho A_x^{1/2}$, $x\in\Omega _A$. Since
\begin{equation*}
\trace\sqbrac{\lscript _x^A(\rho )}=\trace (A_x^{1/2}\rho A_x^{1/2})=\trace (\rho A_x)
\end{equation*}
for all $\rho\in\dscript (H)$, we have $(\lscript _x^A)^\wedge =A_x$ so $\lscript ^A$ measures $A$.

We now consider transition probabilities in the Hilbert space formulation of quantum mechanics. Let $\iscript\in\oscript (H)$, $\rho\in\dscript (H)$, $A,B\in\escript (H)$. If
$\iscript$ has Kraus decomposition $\iscript (\rho )=\sum K_i\rho K_i^*$ with $\sum K_i^*K_i\le I$, then 
\begin{align*}
\trace\sqbrac{\iscript (\rho )}&=\trace\paren{\sum K_i\rho K_i^*}=\sum\trace (K_i\rho K_i^*)=\sum\trace (\rho K_i^*K_i)\\
   &=\trace\paren{\rho\sum K_i^*K_i}
\end{align*}
Assuming that $\trace\sqbrac{\iscript (\rho )}\ne 0$, we have the updated state
\begin{equation*}
\iscript (\rho )^\sim =\frac{\iscript (\rho )}{\trace\paren{\rho\sum K_i^*K_i}}=\frac{\sum K_i\rho K_i^*}{\trace\paren{\rho\sum K_i^*K}}
\end{equation*}
If $\iscript =\lscript ^C$, $C\in\escript (H)$ is a L\"uders operation $\lscript ^C(\rho )=C^{1/2}\rho C^{1/2}$, then
\begin{equation*}
\lscript ^C(\rho )^\sim =\frac{\lscript ^C(\rho )}{\trace\sqbrac{\lscript ^C(\rho )}}=\frac{C^{1/2}\rho C^{1/2}}{\trace (\rho C)}
\end{equation*}

For $A,B\in\escript (H)$ and general $\iscript\in\oscript (H)$, the transition probability of $A$ to $B$ relative to $\rho$, $\iscript$ becomes
\begin{align*}
P_{\rho ,\iscript}(A,B)=\trace (\rho A)\trace\sqbrac{\iscript (\rho )^\sim B}&=\frac{\trace (\rho A)}{\trace\paren{\rho\sum K_i^*K_i}}\,\trace\paren{\sum K_i\rho K_i^*B}\\
   &=\frac{\trace (\rho A)}{\trace\paren{\rho\sum K_i^*K_i}}\,\trace\paren{\rho\sum K_i^*BK_i}
\end{align*}
If $\iscript =\lscript ^C$ is a L\"iders operation, we obtain
\begin{equation*}
P_{\rho ,\lscript ^C}(A,B)=\frac{\trace (\rho A)}{\trace (C^{1/2}\rho C^{1/2})}\,\trace (\rho C^{1/2}BC^{1/2})
   =\frac{\trace (\rho A)}{\trace (\rho C)}\,\trace (C^{1/2}\rho C^{1/2}B)
\end{equation*}
When $\psi$ is a unit vector in $H$, $\rho =\ket{\psi}\bra{\psi}$ is a pure state and we have
\begin{align*}
P_{\rho ,\lscript ^C}(A,B)&=\frac{\trace\paren{\ket{\psi}\bra{\psi}A}}{\trace\paren{\ket{\psi}\bra{\psi}C}}\,\trace\paren{\ket{\psi}\bra{\psi}C^{1/2}BC^{1/2}}\\
   &=\frac{\elbows{A\psi ,\psi}}{\elbows{C\psi ,\psi}}\,\elbows{C^{1/2}BC^{1/2}\psi ,\psi}
\end{align*}

Now let $\rho _1,\rho _2\in\dscript (H)$ and $\iscript\in\oscript (H)$. We then have the transition probability
\begin{equation*}
P_\iscript (\rho _1,\rho _2)=\iscript (\rho _1)(I)\iscript\paren{\iscript (\rho _2)}(I)=\trace\sqbrac{\iscript (\rho _1)}\trace\sqbrac{\iscript\paren{\iscript (\rho _2)}}
\end{equation*}
If $\iscript$ has the above Kraus decomposition, then
\begin{align*}
P_\iscript (\rho _1,\rho _2)&=\trace\paren{\sum _iK_i\rho _1K_i^*}\trace\paren{\sum _{i,j}K_jK_i\rho _2K_i^*K_j^*}\\
   &=\trace\paren{\rho _1\sum _iK_i^*K_i}\trace\paren{\rho _2\sum _{i,j}K_i^*K_j^*K_jK_i}
\end{align*}
In case of a L\"uders operation $\iscript =\lscript ^A$ we obtain
\begin{equation*}
P_{\lscript ^A}(\rho _1,\rho _2)=\trace (\rho _1A)\trace (\rho_2 A^2)
\end{equation*}
When $\rho _1=\ket{\psi _i}\bra{\psi _1}$, $\rho _2=\ket{\psi _2}\bra{\psi _2}$ are pure states, this becomes 
\begin{equation*}
P_{\lscript ^A}(\rho _1,\rho _2)=\trace\paren{\ket{\psi _1}\bra{\psi _1}A}\trace\paren{\ket{\psi _2}\bra{\psi _2}A^2}
   =\elbows{A\psi _1,\psi _1}\elbows{A^2\psi _2,\psi _2}
\end{equation*}
If $A$ is a projection, we have
\begin{equation*}
P_{\lscript ^A}(\rho _1,\rho _2)=\elbows{A\psi _1,\psi _1}\elbows{A\psi _2,\psi _2}
\end{equation*}
and when $A=\ket{\psi}\bra{\psi}$ is a one-dimensional projection, this becomes
\begin{equation*}
P_{\lscript ^A}(\rho _1,\rho _2)=\ab{\elbows{\psi ,\psi _1}}^2\ab{\elbows{\psi ,\psi _2}}^2
\end{equation*}
Finally, if $\psi =\psi _1$ we obtain $P_{\lscript ^A}(\rho _1,\rho _2)=\ab{\elbows{\psi _1,\psi _2}}^2$ which is the usual transition probability for pure states. We conclude that the usual transition probability is a very special case.

\end{document}